\documentclass[aps,pra,11pt,superscriptaddress,notitlepage,bibnotes,nofootinbib]{revtex4-1}

\usepackage{amsthm,amsmath,amssymb,amsbsy}

\usepackage{hyperref}
\usepackage{url}

\usepackage[normalem]{ulem}

\usepackage{graphicx}
\usepackage{bm,bbm}%
\usepackage{braket}
\usepackage{enumerate}
\usepackage{booktabs,multirow}
\usepackage{mathtools,bbding}
\usepackage{microtype}
\usepackage{mathrsfs}

\usepackage{txfonts,newtxmath}

\usepackage[most,breakable]{tcolorbox}

\newcommand{\nc}{\newcommand}
\newcommand{\rnc}{\renewcommand}

\makeatletter
\newcommand*\rel@kern[1]{\kern#1\dimexpr\macc@kerna}
\newcommand*\widebar[1]{%
  \begingroup
  \def\mathaccent##1##2{%
    \rel@kern{0.8}%
    \overline{\rel@kern{-0.8}\macc@nucleus\rel@kern{0.2}}%
    \rel@kern{-0.2}%
  }%
  \macc@depth\@ne
  \let\math@bgroup\@empty \let\math@egroup\macc@set@skewchar
  \mathsurround\z@ \frozen@everymath{\mathgroup\macc@group\relax}%
  \macc@set@skewchar\relax
  \let\mathaccentV\macc@nested@a
  \macc@nested@a\relax111{#1}%
  \endgroup
}

\rnc{\thesection}{\arabic{section}}
\rnc{\thesubsection}{\thesection.\arabic{subsection}}
\rnc{\thesubsubsection}{\thesubsection.\arabic{subsubsection}}

\usepackage{etoolbox}
\hypersetup{linktoc=all}

\newtheorem{definition}{Definition}

\newtheorem{theorem}{Theorem}
\newtheorem*{theorem*}{Theorem}

\DeclareMathOperator{\Tr}{Tr}

\nc{\tcr}[1]{{\color{red} #1}}
\nc{\tcb}[1]{{\color{blue} #1}}

\nc{\psic}{\psi^{c}}
\nc{\two}[1]{\underline{2^{d-#1}}}
\rnc{\H}{\mathcal{H}}
\nc{\Hanc}{\mathcal{H}_{\text{anc}}}
\nc{\psianc}{\psi_{\text{anc}}}
\nc{\lampow}{\lambda^{1/d}}
\nc{\tl}[1]{\tilde{#1}}
\nc{\CHSH}{\text{CHSH}}

\nc{\inner}[1]{\langle #1 \rangle}
\nc{\proj}[1]{\ket{#1}\!\bra{#1}}
\nc{\pro}[1]{#1 #1^\dagger}
\nc{\RR}{{{\mathbb R}}}
\nc{\CC}{{{\mathbb C}}}
\nc{\FF}{{{\mathbb F}}}
\nc{\NN}{{{\mathbb N}}}
\nc{\ZZ}{{{\mathbb Z}}}

\nc{\MIO}{{\text{\rm MIO}}}
\nc{\DIO}{{\text{\rm DIO}}}
\nc{\SIO}{{\text{\rm SIO}}}
\nc{\IO}{{\text{\rm IO}}}

\nc{\SEP}{{\text{SEP}}}
\nc{\NS}{{\text{NS}}}
\nc{\LOCC}{{\text{LOCC}}}
\nc{\PPT}{{\text{PPT}}}
\nc{\EXT}{{\text{EXT}}}
\nc{\OLOCC}{{\text{1-LOCC}}}
\nc{\SEPP}{{\text{SEPP}}}

\nc{\MC}{{\text{\rm MC}}}

\nc{\cE}{\mathscr{E}}

\nc{\nn}{\nonumber}

\rnc{\*}{\textup{*}}
\nc{\ketbra}[1]{\ket{#1}\!\!\bra{#1}}

\newcommand{\N}{\mathcal{N}}

\renewcommand{\P}{\mathcal{P}}

\nc{\MM}{\widetilde{\M}}
\nc{\Ml}{\M^{\leq}}

\nc{\mleq}{\preceq}
\nc{\mgeq}{\succeq}

\nc{\ox}{\otimes}

\nc{\wt}{\widetilde}

\nc{\SDP}{\text{\rm SDP}}

\nc{\cc}{{\circ\circ}}
\nc{\mnorm}[1]{\norm{#1}{[m]}}

\nc{\F}{\mathcal{F}}
\nc{\M}{\mathcal{M}}

\let\oldproofname\proofname
\rnc{\proofname}{\rm\bf{\oldproofname}}
\rnc{\qedsymbol}{{\color{gray!50!black}\rule{0.6em}{0.6em}}}

\newcommand{\bb}{\begin{equation}}
\newcommand{\bbb}{\begin{equation*}}
\newcommand{\ee}{\end{equation}}
\newcommand{\eee}{\end{equation*}}

\nc{\note}[1]{{\color{blue!90!black} #1}}
\usepackage[hmargin=2cm,vmargin=2.5cm]{geometry}
\usepackage{sidecap}

\renewcommand*\env@matrix[1][*\c@MaxMatrixCols c]{%
  \hskip -\arraycolsep
  \let\@ifnextchar\new@ifnextchar
  \array{#1}}

\begin{abstract}
	{Bell’s theorem shows that no hidden-variable model can explain the measurement statistics of a quantum system shared between two parties, thus ruling out a classical (local) understanding of nature. In this work we demonstrate that by relaxing the positivity restriction in the hidden-variable probability distribution it is possible to derive {quasiprobabilistic} Bell inequalities whose sharp upper bound is written in terms of a negativity witness of said distribution. 
		This provides an analytic solution for the amount of negativity necessary to violate the CHSH inequality by an arbitrary amount, therefore revealing the amount of negativity required to emulate the quantum statistics in a Bell test.  }
\end{abstract}
\begin{document}
	\title{\Large Witnessing Bell violations through probabilistic negativity}
	\author{Benjamin Morris} 
	\email{morris.quantum@gmail.com} 
	\affiliation{\small School of Physics and Astronomy and Centre for the Mathematics and Theoretical Physics of Quantum Non-Equilibrium Systems,
		University of Nottingham, University Park, Nottingham NG7 2RD, United Kingdom}
	\author{Lukas J.~Fiderer}
	\affiliation{\small Institute for Theoretical Physics, University of Innsbruck, 6020 Innsbruck, Austria}
	\author{Ben Lang}
	\affiliation{\small School of Physics and Astronomy and Centre for the Mathematics and Theoretical Physics of Quantum Non-Equilibrium Systems,
		University of Nottingham, University Park, Nottingham NG7 2RD, United Kingdom} 
	\author{Daniel Goldwater} 
	\affiliation{\small School of Mathematical Sciences and Centre for the Mathematics and Theoretical Physics of Quantum Non-Equilibrium Systems,
		University of Nottingham, University Park, Nottingham NG7 2RD, United Kingdom}
	\maketitle

	\section{Introduction} 
	It has now been 60 years since John Stewart Bell wrote his famous paper on the {Einstein}-Podolsky-Rosen (EPR) paradox \cite{bell1964einstein}, and 50 years since the first experimental Bell test \cite{freedman1972experimental}. The majority of physicists are perfectly happy to concede that in the lab we see experimental results consistent with the postulates of quantum mechanics. However, the implications of these mathematical postulates on the `reality' of the wavefunction is still very much up for debate \cite{caves2002quantum,harrigan2010einstein,pusey2012reality,colbeck2012system,mermin2014physics,ringbauer2015measurements}.

	These Bell experiments remain as some of the most important demonstrations for the reality of the quantum state and the death of a `local realism' picture of nature. In such an experiment a physical system is distributed between spatially separated observers, and we allow these observers to perform measurements on their local system. The emerging statistics prove that physical systems are not bound to behave locally (in accordance to local hidden-variable models). Rather, the statistics are consistent with the postulates governing quantum {mechanics}.
	
	In this work we remove the postulates of quantum mechanics and instead allow a physical system to be distributed according to a {quasiprobability (hidden-variable) distribution that is allowed to take negative values}. Although we are perfectly content with real negative numbers in physics, negative {quasi}probabilities-- despite receiving support from individuals such as Dirac \cite{dirac1942bakerian} and Feynman \cite{feynman1987negative} and having a solid mathematical foundation \cite{ruzsa1988algebraic, khrennikov2007generalized}-- have been a long debated issue in theoretical physics \cite{muckenheim1986review}. See, for example, the extensive discussion surrounding the interpretation of negative values in the Wigner distribution \cite{ferrie2011quasi,veitch2012negative}. {{ In the majority of considerations, quasiprobability distributions are used to describe states that are not directly observed}; that is, all observable measurement statistics must be governed by ordinary probability distributions. {As an example{,} a Wigner function may assign a negative quasiprobability to a particle having a particular position/momentum combination, but any physical measurement, constrained by Heisenberg uncertainty, will have an all-positive outcome distribution.} This feature ensures that no outcome is ever predicted to be seen occurring a negative number of times \cite{feynman1987negative}, and similarly protects the quasiprobability physicist from falling victim to `Dutch book' arguments \cite[Ch.3]{de2017theory}}.
	
	An important motivator for this work is the result of Al-Safi and Short \cite{al2013simulating} (expanded upon by the authors of \cite{oas2014exploring}) which showed that it is possible to simulate all non-signalling correlations, (those which adhere to the principles of {special }relativity) \cite{popescu1994quantum,peres2004quantum}, if one allows negative values in a probability distribution. However, physical reality does not explore this full set of correlations -- but rather, is restricted to those achievable by quantum correlations. Therefore the question that we pose in this work is:\begin{quote}\center
		``\textit{What are the restrictions on the negativity in a {hidden-variable probability distribution} such that it can emulate the statistics seen in a physical Bell experiment?}'' 
	\end{quote}  \vspace{1em}
	In order to answer this question we  {construct
	CHSH inequalities} for two parties \cite{clauser1969proposed} whose degree of violation is witnessed by the amount of negativity present in the hidden-variable probability distribution. Our witness { yields}  a value of 0 for a quasiprobability distribution which is entirely positive, such as that which
	would describe an ordinary classical system.

	We start by describing the setup necessary for the construction of these nonlocal experiments, introducing the probability distributions admitted by classical, quantum, and non-signalling {theories}, and giving the famous Bell scores that these distributions respectively allow one to reach in such nonlocal experiments.  We then give the definition of a quasiprobability distribution and motivate \textit{negativity witnesses} as a quantitative method of detecting negativity in said distributions. Our main result is that the violation of the CHSH inequality (and $n$-measurement generalisations) can be exactly characterised by a negativity witness of the hidden-variable distribution defined over the local states, and that there exists quasiprobability distributions which can saturate (up to the no-signalling limit) any such violation, whilst still having well-defined local statistics. This shows that it is possible to {recapture} the nonlocal features of {Bell experiments} through having a finite amount of negativity allowed in a hidden-variable distribution over scenarios which are, in themselves, entirely local and classical.    
	\begin{SCfigure*}
	\centering
		\includegraphics[width=0.6\textwidth]{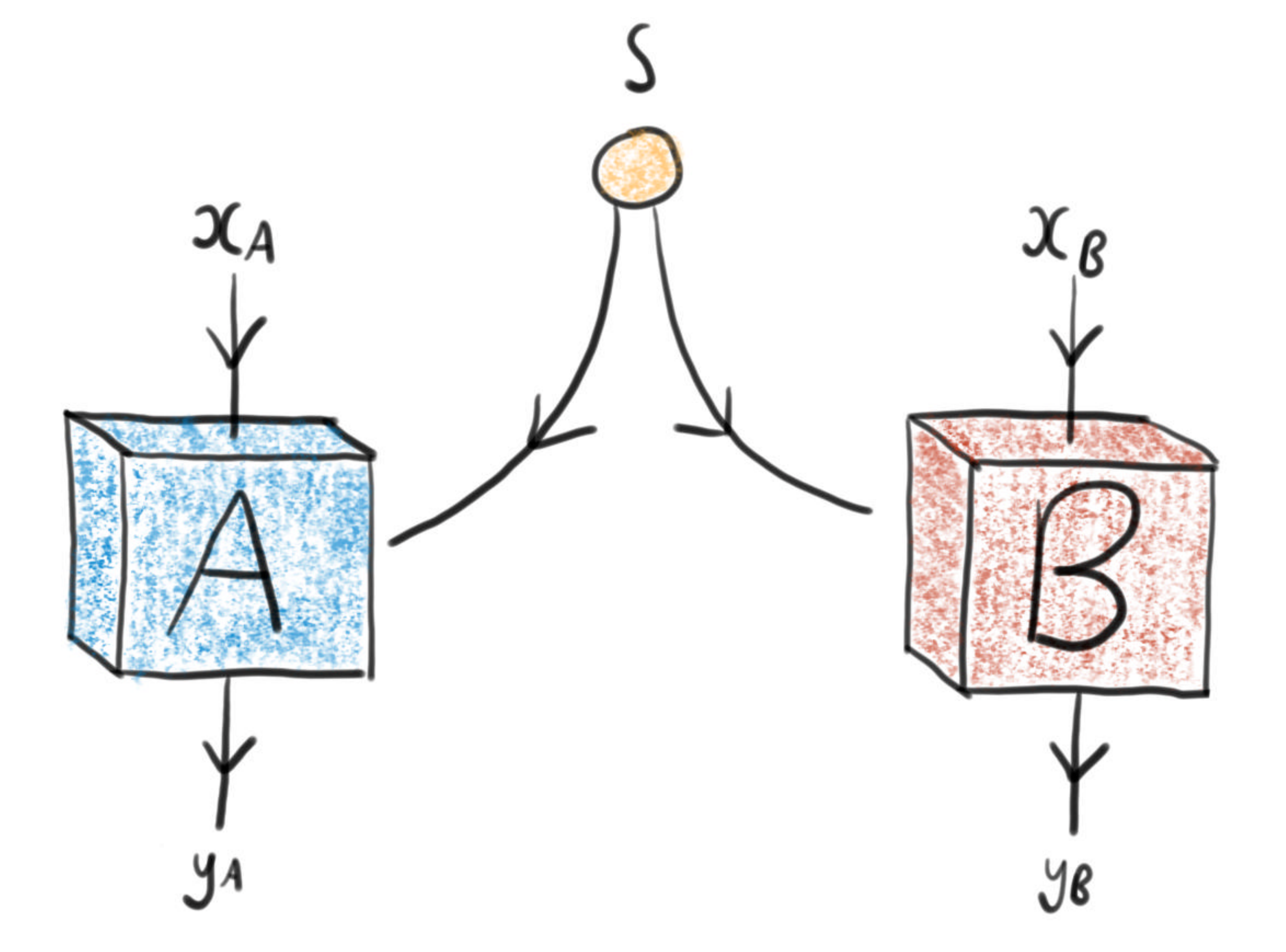}
	\caption{A source $S$ distributes a system between two spatially separated observers, Alice ($A$) and Bob ($B$). Alice and Bob choose to measure their local part of the system with measurements $x_A,x_B$, with possible outcomes  {$y_{A},y_B\in\{-1,+1\}$}. The statistics of their measurement outcomes depending upon the physical system being distributed by the source.     \label{fig:box1}}
	\end{SCfigure*}
		\section{Setup}
	{Let us consider the following experimental setup. A source $S$ distributes a system between $2$ observers, the $k$th observer can choose some measurement $x_k\in\{0,1,\dots,X_k\}$ and record some outcome $y_k\in\{1, 2,\dots,Y_k\}$, the possible values of $k$ being $\{A,B\}$. A specific experimental {setup} is characterised by the conditional probability, 
		\begin{align}
		P\left(y_A,y_B|x_A,x_B\right).
		\end{align}
		{The physical theory} governing the behaviour of the system and experiment determines the achievability of certain conditional probability distributions resulting from these experiments. We are interested in the following three physical theories:}

	$\bullet$ \textit{Classical theory} admits probability distributions of the following form, 
	\begin{align}
	P\left(y_A,y_B|x_A,x_B\right)=\sum_{\lambda_A,\lambda_B}P_A\left(y_A|x_A,\lambda_A\right)P_B\left(y_B|x_B,\lambda_B\right)P_\Lambda\left(\lambda_A,\lambda_B\right),\label{eq:classical}
	\end{align}
	where $P_\Lambda\left(\lambda_A,\lambda_B\right)$ is {a} joint probability distribution defined over local hidden variables. {With each choice of hidden variables we associate a local scenario governed by ordinary local probability distributions $P_k\left(y_k|x_k,\lambda_k\right)$ for the observables $y_A,y_B$. The hidden-variable probability distribution $P_\Lambda$ determines how such local scenarios are mixed, and the probability distributions $P_k\left(y_k|x_k,\lambda_k\right)$ are called ``$\lambda_k$-local'' because they belong to the scenario associated with a particular value of $\lambda_k$, not to be confused with the (observable) marginal probability distributions that are obtained by marginalising the total probability distribution, equation \eqref{eq:classical}.}

	{$\bullet$ \textit{Quantum theory} endows us with a Hilbert space structure for our quantum states that  admits probability distributions of the following form, \begin{align}
		P\left(y_A,y_B|x_A,x_B\right)=\Tr\left[\left(M^{(A)}_{y_A|x_A}\ox M^{(B)}_{y_B|x_B}\right)\rho\right],\label{eq:prob_quantum}
		\end{align}
		where $\rho\succeq 0$ and $M^{(k)}_{y_k|x_k}$ are Positive Operator-Valued Measures (POVMs) \cite{2000quantum} for each $k$. }
	
	{$\bullet$ \textit{No-signalling principle}, our third physical theory, prohibits the sending of information faster than the speed of light \cite{popescu1994quantum,peres2004quantum}. Such a theory has the conditions on its conditional probability distribution that for any {$k\in \{A,B\}$}, 
		\begin{align}
		\sum_{y_k} P\left(y_A,y_B|x_A,x_B\right)\label{eq:cond_prob}
		\end{align} is independent of $x_k$.} 
	These three physical theories ranging from the most restrictive(classical), to the least restrictive (no-signalling), with quantum theory existing somewhere between the two \cite{popescu1994quantum}.
	
	{Representing the full set of correlations that the quantum conditional probability distribution in equation \eqref{eq:prob_quantum} allows one to reach is a notorious problem, and the set has recently been shown to be not closed \cite{slofstra2019set}. Therefore we instead restrict ourselves to studying the achievable bounds that these conditional probabilities allow {one} to reach in nonlocal experiments {; the} original and most famous of which being the Bell inequality \cite{bell1964einstein}.
		\begin{definition}[Bell inequality] Given observers {$A$ and $B$}, each with measurement choice $x_k\in\{0_k,1_k\}$ with outcomes $y_k\in\{-1,+1\}$, experiments performed on the systems adhere to the bound,
			\begin{align}
			\left|E(0_A,0_B)-E(0_A,1_B)+E(1_A,0_B)+E(1_A,1_B)\right|\leq X,\label{equ:all_bell}
			\end{align}
			where both the correlation measure $E(x_A,x_B)=\sum_{y_A,y_B}y_Ay_B\, P(y_A,y_B|x_A,x_B)$ and bound $X\in\mathbb{R}^+$ are theory {dependent}. {The left-hand side of that inequality is often called the score of the experiment.}
			\label{def:Bell}
	\end{definition}}
	Each physical theory admits a different conditional probability distribution, and hence a different achievable bound $X$. 
	{Classical theory has the CHSH bound $X=2$ \cite{clauser1969proposed}, quantum theory has the Tsirelson bound of $X=2\sqrt{2}$ \cite{cirel1980quantum} and non-signalling distributions $X=4$ \cite{popescu1994quantum}.} We are interested in the achievable bounds of a classical system's probability distribution when the hidden-variable distribution in said probability distribution can be negative.
	
	\section{Results}

	We now define an important object for this work, the \textit{quasiprobability distribution}.
	\begin{definition}[Quasiprobability distribution]\label{def:quasi_prob}
		We define a quasiprobability distribution as {$\tl{P}_\Lambda:\Lambda_1\times\dots\times\Lambda_N\rightarrow\mathbb{R}$} where {$\Lambda_i\subset \mathbb{R}$ and } $|\Lambda_i|<\infty\,\forall \,i$, that is properly normalised, such that,
		\begin{align}
		\sum_{\lambda_1,\dots,\lambda_N}\tl{P}_\Lambda\left(\lambda_1,\dots,\lambda_N\right)=1.
		\end{align}
	\end{definition}
	It can be seen that the collection of functions adhering to the above definition forms a convex set, which we will denote $\tl{\P}$, a super set of the convex set of positive probability distributions $\P\subset\tl{\P}$.
	%%%%%%%%%%%%%%%%%%%%%%%%%%%%%%%%%%%%%%%%%%%%%%%%%%%%%%%%%%%%%%%%%%%%%%%%%%%%%
	We must now determine how to quantify the presence of negativity in our {quasiprobability distributions}. To this end we will use a well{-}known method for quantitatively detecting properties of a quantum state, witnesses \cite{terhal2000bell,lewenstein2000optimization,eisert2007quantitative,chabaud2021witnessing}. Let us therefore proceed by defining a \emph{negativity witness:} 
	\begin{definition}[Negativity witness]\label{def:neg_witness} Given some properly normalised probability distribution $P$, a well-defined \textit{negativity witness} is one which, 
		\begin{align}
		\mathcal{N}(P)=0\,\,\forall\,P\in{\P}.
		\end{align}
		We may additionally require such a witness to `faithfully' detect negativity, 
		\begin{align}
		\mathcal{N}(P)>0\,\,\forall\,P\in \tilde{\P}{\backslash}\P.
		\end{align}
	\end{definition}

	In the following, we consider classical local hidden-variable models{ as defined in } equation \eqref{eq:classical}, but we replace the hidden-variable probability distribution $P_\Lambda$ with a quasiprobability distribution $\tilde{P}_\Lambda$,
		\begin{align}
		P\left(y_A,y_B|x_A,x_B\right)=\sum_{\lambda_A,\lambda_B}P_1\left(y_A|x_A,\lambda_A\right)P_2\left(y_B|x_B,\lambda_B\right)\tilde{P}_\Lambda\left(\lambda_A,\lambda_B\right).\label{eq:quasi_hidden}
		\end{align}
		
		This corresponds to a scenario where different local statistics of observations, governed by {$\lambda$-local (ordinary)} probability distributions $P_k\left(y_k|x_k,\lambda_k\right)$, are mixed according to a quasiprobability distribution $\tilde{P}_\Lambda$.
		{However, when $\tilde{P}_\Lambda$ takes negative values, we should no longer think of the model as an ignorance mixture of valid local scenarios but rather as a nonlocal model \cite{al2013simulating}. Furthermore, when compared with ordinary hidden-variable models not all combinations of hidden-variable and $\lambda$-local probability distributions are valid; only those combinations  which lead to well-defined $	P\left(y_A,y_B|x_A,x_B\right)$ {are valid}, i.e., comprised of values between 0 and 1 (the normalisation condition is always fulfilled).} 
		
 In addition to the correlation function  between {two measurements $x_A$ and $x_B$,
${E}(x_A,x_B)$, it will also be useful to define $\lambda_k$-local expectation values {corresponding to an imagined scenario where} observer $k$ {is able to perform} measurement $x_k$ in {the} local scenario {corresponding to} $\lambda_k$,
		\begin{align}
		\inner{k}_{\lambda_k}^{x_k}:=\sum_{y_k}y_k P_k(y_k|x_k,\lambda_k). \label{eq:local_exp}
		\end{align}
		This $\lambda_k$-local expectation value will be useful to formulate our results, but does not correspond to the actual observations which are themselves governed by equation \eqref{eq:quasi_hidden}. }
	
	We are now in a position to state the main result of this work{, the quasiprobabilistic Bell inequality}.
	\begin{theorem}[{Quasiprobabilistic Bell inequality}]\label{thm:CHSH_break}
		Given observers {$A$ and $B$}, each with measurement choice $x_k\in\{0_k,1_k\}$ with outcomes $y_k\in\{-1,+1\}$ whose systems are distributed according to some  {quasi}probability  distribution $\tilde{P}_\Lambda$, then the quasiprobabilistic Bell inequality holds:
		\begin{align}
		\left|{E}(0_A,0_B)-{E}(0_A,1_B)+{E}(1_A,0_B)+{E}(1_A,1_B)\right|\leq\,2+\mathcal{N}(\tilde{P}_\Lambda),\label{equ:quasi_CHSH2}
		\end{align}
		\color{black}
		{where}
		{\begin{align}
			\N(\tilde{P}_\Lambda)\coloneq\begin{cases}
			\N_+(\tilde{P}_\Lambda) \qquad\text{ if }{E}(1_A,0_B)+{E}(1_A,1_B)<0,\\
			\N_-(\tilde{P}_\Lambda) \qquad\text{ else,}
			\end{cases}
			\end{align}
		{is a negativity witness, and }
			 $$\N_\pm(\tilde{P}_\Lambda):=\sum_{\lambda_A,\lambda_B}\left[2\pm\left(\inner{A}_{\lambda_A}^{1_A}\inner{B}_{\lambda_B}^{1_B}+\inner{A}_{\lambda_A}^{1_A}\inner{B}_{\lambda_B}^{0_B}\right)\right]\left(\Big{|}\tilde{P}_\Lambda(\lambda_A,\lambda_B)\Big{|}-\tilde{P}_\Lambda(\lambda_A,\lambda_B) \right).$$  }
	\end{theorem}
	The proof of this theorem begins analogously with Bell's proof of the CHSH bound \cite{bell2004speakable}, but diverges when the assumption $P\in\P$ is {made in Bell's proof}. The above result shows that if an arbitrary amount of negativity is allowed in  {the hidden-variable probability distribution} then the upper bound of equation \eqref{equ:quasi_CHSH2}  can be arbitrarily large. However, it should be noted that a natural limit of $4$ in the relevant Bell tests (i.e., for the upper bound in the quasiprobabilistic Bell inequality) is imposed by the requirement that $	P\left(y_A,y_B|x_A,x_B\right)$ is a well-defined, valid probability distribution \footnote{This can be easily seen as for {well-defined experimental setups where $y_k=\{-1,+1\}$, i.e.~for valid} $	P\left(y_A,y_B|x_A,x_B\right)$, each term on the left hand side of equation \eqref{equ:quasi_CHSH2} can be at most $1$ and least $-1$. }. 
	
	The previous result of Al-Safi and Short \cite{al2013simulating} {showed} that it was possible to violate said inequality up to this no-signalling bound of $X=4$. Therefore, in order to emulate the physical results seen in Bell test{s} (Tsirelson bound) one needs a negative probability distribution whose witness equals $\mathcal{N}(\tilde{P}_\Lambda)=2(\sqrt{2}-1)$. In section \ref{sec:examples} we show that for any $\mathcal{N}({\tilde{P}_\Lambda}){\leq 2}$, there exist {quasiprobabilistic hidden-variable models with valid local measurement statistics that}  saturate inequality \eqref{equ:quasi_CHSH2}. We would hope that if a physical mechanism was discovered that allowed a joint hidden-variable probability distribution to have the appearance of negativity, one would expect that said physical mechanism was limited in such a way that it resulted in the Tsirelson bound and more generally was able to reconstruct the limits on quantum correlations.

	It is also important to note that although said witness $\N(\tilde{P}_\Lambda)$ is a valid witness according to definition \ref{def:neg_witness} it is not necessarily a `faithful' one. However this can be rectified, at the cost of loosening the bound, by redefining said witness. For example the function {$\mathcal{N}'({\tilde{P}_\Lambda}):=\sum_{\lambda_A,\lambda_B}4\left(\Big{|}\tilde{P}_\Lambda(\lambda_A,\lambda_B)\Big{|}-\tilde{P}_\Lambda(\lambda_A,\lambda_B) \right)$}, is defined to be both a valid and `faithful' witness.
	
	There are numerous generalisations of the famous CHSH inequalities, {such as} multiple parties \cite{perez2008unbounded}, arbitrary outcomes \cite{cope2019bell}, etc. These would no doubt be interesting to study but we leave it to future work to explore these other generalisations and instead focus on the scenario in which Alice and Bob have access to an arbitrary number of measurement settings \cite{wehner2006tsirelson}.
	\begin{theorem}\label{thm:CHSH_break_N}
		Given observers {$A$ and $B$}, each with $n\geq2$ measurements 
		$x_k\in\{0_k,1_k,\dots,{n-1}_k\}$ with outcomes $y_k\in\{-1,+1\}$ whose systems are distributed according to some  {quasi}probability distribution $\tilde{P}_\Lambda$, 
		{\begin{align}
			\left|\sum^{n-1}_{i=0}{E}(i_A,i_B)+\sum^{n-1}_{i=1}{E}(i_A,i-1_B)-{E}(0_A,n-1_B)\right|\leq\,2n-2+{\mathcal{N}_n({\tilde{P}_\Lambda})}, \label{eq:CHSH_break_N}
			\end{align}}
		where { $\mathcal{N}_n({\tilde{P}_\Lambda})=\sum^{n-1}_{i=1}\N^{(i)}(\tilde{P}_\Lambda)$ {is a negativity witness} with}
		\begin{align}
		\N^{(x)}(\tilde{P}_\Lambda)\coloneq\begin{cases}
		\N^{(x)}_+(\tilde{P}_\Lambda) \qquad\text{ if }{E}(0_A,x_B)+{E}(0_A,x-1_B)<0,\\
		\N^{(x)}_-(\tilde{P}_\Lambda) \qquad\text{ else,}
		\end{cases}\label{eq:neg_cases_n}
		\end{align}
		where $\N^{(x)}_{\pm}({\tilde{P}_\Lambda})\coloneqq\sum_{\lambda_A,\lambda_B}\left[2\pm\left(\inner{A}_{\lambda_A}^{x_A}\inner{B}_{\lambda_B}^{x_B}+\inner{A}_{\lambda_A}^{x_A}\inner{B}_{\lambda_B}^{x-1_B}\right)\right]\left(\Big{|}{\tilde{P}_\Lambda}(\lambda_A,\lambda_B)\Big{|}-{\tilde{P}_\Lambda}(\lambda_A,\lambda_B) \right)$.
	\end{theorem}
	The proof of the above theorem can be found in appendix \ref{app:proof_thm2}, it utilises proof by induction by chaining together the inequalities from theorem \ref{thm:CHSH_break}. In section \ref{sec:examples} we show that the bound in theorem \ref{thm:CHSH_break_N} can be saturated. Namely, {for any $\mathcal{N}_n(\tilde{P}_\Lambda){\leq 2}$, there exist well-defined $	P\left(y_A,y_B|x_A,x_B\right)$, characterised by a  quasiprobability hidden-variable distribution $\tilde{P}_\Lambda(\lambda_A,\dotsc,\lambda_n)$, that saturate inequality \eqref{eq:CHSH_break_N}.} In addition, analogously to the two measurement result, at the cost of loosening the bound we can ensure that the above witness is also `faithful' by  choosing for all $x$, $\mathcal{N}^{(x)}(\tilde{P}_\Lambda)=\mathcal{N}'({\tilde{P}_\Lambda})$. 
	%%%%%%%%%%%%%%%%%%%%%%%%%%%%%%%%%%%%%%%%%%%%%%%%%%%%%%%%%%%%%%%%%%%%%%%%%%
	\section{Example}\label{sec:examples}
	In order to understand how to saturate the Bell inequality from Theorem \ref{thm:CHSH_break}, we rewrite the left-hand side of equation \eqref{equ:quasi_CHSH2} as,
	\begin{align}
\Bigg{|}\sum_{\lambda}\M(\lambda)\,\tilde{P}_\Lambda(\lambda)\Bigg{|}. \label{eq:CHSH_example2}
	\end{align}
We replaced the hidden variables $\lambda_A$ and $\lambda_B$ with a single hidden variable $\lambda$ because our example only uses a single hidden variable $\lambda$. Further,   $\M (\lambda):=\inner{A}_{\lambda}^{0_A}\inner{B}_{\lambda}^{0_B}-\inner{A}_{\lambda}^{0_A}\inner{B}_{\lambda}^{1_B}+\inner{A}_{\lambda}^{1_A}\inner{B}_{\lambda}^{0_B}+\inner{A}_{\lambda}^{1_A}\inner{B}_{\lambda}^{1_B}$ are the scores of each of the $\lambda$-local distributions {, that is }  $-2\leq \M(\lambda)\leq 2$ {holds}.
	
	For a given value of the negativity witness, we exceed the local bound maximally by the simple strategy of weighting classical distributions with $\M (\lambda)=+2$ with positive quasiprobability, while simultaneously taking a classical distribution with $\M (\lambda)=-2$ with {negative weight}.
    To ensure that the {total}  probability distribution $P\left(y_A,y_B|x_A,x_B\right)$ is well-defined we {make} a choice of three deterministic classical distributions with positive weight and a fourth with negative weight. Our four deterministic classical distributions can be denoted $[(-,-)_A,(-,+)_B], [(+,-)_A,(-,-)_B], [(+,+)_A,(+,+)_B]$ and $[(+,-)_A,(-,+)_B]$. { Here, our notation means} that the distributions can be produced by assigning the first pair of symbols to Alice and the second to Bob. Each party chooses to read either the first or second of the symbols given to them (this choice reflects their measurement setting $x_k$) while the outcome of their measurement is determined by the symbol itself; that is{,} $y_k=+1$ ($y_k=-1$) for a plus (minus) sign. This {experimental description of distributing classical information makes clear that these distributions are local}{,} with our hidden variable $\lambda$ indicating which of these sets the source actually produces.
    
    The source produces each of the distributions according to the following quasiprobability distribution, 
         \begin{align}\label{eq:prob_sat_2}
         \tilde{P}_\Lambda(\lambda)=\begin{cases}
 	\frac{4+\mathcal{N}}{12}\qquad\text{ for }\lambda=1,2,3,\\
 	-\frac{\mathcal{N}}{4} \qquad\,\,\,\text{  for }\lambda=4,
 		\end{cases}
     \end{align}
    where{ $\M (\lambda)=2$ if $\lambda=1,2,3$  and  $\M (\lambda)=-2$ if $\lambda=4$}.  We can use tables to represent $\lambda$-local probability distributions, and the total probability distribution is then given as the weighted sum of such tables:
	%%%%%%%%%%%%%%%%%%%%%%%%%%%%%%%%%%%%%%%%
\begin{align}\label{CHSH_table}
&\scriptsize	\frac{4 + \mathcal{N}}{12} \begin{pmatrix}                    \begin{bmatrix}[c|cccc]
		     & \multicolumn{4}{c}{y_A y_B} \\
			x_A x_B & -- & -+ & +- & ++ \\ \hline
			00        & 1 & 0 & 0 & 0 \\
			01        & 0 & 1 & 0 & 0 \\
			10        & 1 & 0 & 0 & 0 \\
			11        & 0 & 1 & 0 & 0 \\
		\end{bmatrix}&+\begin{bmatrix}[c|cccc]
		     & \multicolumn{4}{c}{y_A y_B} \\
			x_A x_B & -- & -+ & +- & ++ \\ \hline
			00        & 0 & 0 & 1 & 0 \\
			01        & 0 & 0 & 1 & 0 \\
			10        & 1 & 0 & 0 & 0 \\
			11        & 1 & 0 & 0 & 0 \\
		\end{bmatrix}+\begin{bmatrix}[c|cccc]
		     & \multicolumn{4}{c}{y_A y_B} \\
			x_A x_B & -- & -+ & +- & ++ \\ \hline
			00        & 0 & 0 & 0 & 1 \\
			01        & 0 & 0 & 0 & 1 \\
			10        & 0 & 0 & 0 & 1 \\
			11        & 0 & 0 & 0 & 1 \\
		\end{bmatrix}
	 \end{pmatrix}	
	\scriptsize	- \frac{\mathcal{N}}{4}			\begin{bmatrix}[c|cccc]
		     & \multicolumn{4}{c}{y_A y_B} \\
			x_A x_B & -- & -+ & +- & ++ \\ \hline
			00        & 0 & 0 & 1 & 0 \\
			01        & 0 & 0 & 0 & 1 \\
			10        & 1 & 0 & 0 & 0 \\
			11        & 0 & 1 & 0 & 0 \\
		\end{bmatrix}
		\nonumber\\ 	\nonumber\\ 
	 &\scriptsize\qquad\quad\,\,	=\frac{1}{12}\begin{bmatrix}[c|cccc]
		     & \multicolumn{4}{c}{y_A y_B} \\
			x_A x_B & -- & -+ & +- & ++ \\ \hline
	00 & 4+\mathcal{N} & 0 & 4-2\mathcal{N} & 4+\mathcal{N} \\
	01 &	     0 & 4+\mathcal{N} & 4+\mathcal{N} &  4-2\mathcal{N} \\
	10 &	     8 - \mathcal{N} & 0 & 0 & 4+\mathcal{N} \\
	11 &		 4+\mathcal{N} & 4-2\mathcal{N} & 0 & 4+\mathcal{N} \\
		\end{bmatrix}.
	\end{align}
The requirement that the resulting total probability distribution must be valid implies
 $\N\leq2$ {which corresponds to the no-signalling limit.} {Furthermore, it is easy to check that said distribution indeed gives a value of $\mathcal{N}$ for the negativity witness.}

The quasiprobabilistic Bell inequality score for this experiment is $ 2+\mathcal{N}$, which upon substituting equation \eqref{eq:prob_sat_2} into the negativity witness, can be seen to saturate the bound. In appendix \ref{app:sat_n_ex} we discuss how one can generalise the above to the $n$-measurement scenario. 
	%%%%%%%%%%%%%%%%%%%%%%%%%%%%%%%%%%%%%%%%%%%%%%%%%%%%%%%%%%%%%%%%%
	\section{Conclusion}
	We have shown that there exists a relationship between the amount of negativity allowed in a joint hidden-variable distribution, and the degree to which said distribution can demonstrate nonlocality in a Bell experiment. In particular, theorem  \ref{thm:CHSH_break_N} {introduces a quasiprobabilistic Bell inequality,} which gives us a sharp bound in the scenario of two parties with {$n$} inputs (corresponding to a choice between $n$ measurements) and can be used straightforwardly to reconstruct quantum statistics using nothing more than local, separable classical probability distributions and a quasiprobability distribution over them -- granted an appropriately well spent budget of negativity. 
	
	Our work sits within the long-established tradition of trying to understand quantum theory through interpretative lenses which remove some particular aspect from a classical worldview. Such approaches are wide and varied, including superdeterminism \cite{hossenfelder2020rethinking, adlam2018quantum}; retro-causality \cite{price2012does}; invoking an irreducible role for subjectivity in physics \cite{mermin2014physics, fuchs2014introduction, mueller2020law}; taking physical reality to consist of interacting, separate realms \cite{goldstein2001bohmian, esfeld2014ontology}; allowing the relativity of pre and post-selection \cite{bacciagaluppi2020reverse}; taking Hilbert space to be literal \cite{carroll2021reality}, and so on. Here we add to this list, in that we present an additional way to re-capture the nonlocal features of quantum theory: through having a finite amount of negativity allowed in a hidden-variable distribution over scenarios which are, in themselves, entirely local and classical.  We are not claiming that such quasi distributions are `real' -- only, more modestly, that such a perspective could not be ruled out at this stage. 
	
	Pursuing this line of reasoning, we would hope that our results may help to determine the fundamental restrictions on a system's quasiprobability hidden-variable distribution such that it captures the full character of  physical correlations. Put another way; we know that zero negativity can capture the set of classical correlations, whilst un-bounded negativity can capture the non-signalling set. Given that the set of quantum correlations lies between these two -- what are the restrictions on the quasiprobability hidden-variable distribution which would suffice to identify the full set of quantum correlations? We hope to explore this question in further work. 
	%%%%%%%%%%%%%%%%%%%%%%%%%%%%%%%%%%%%%%%%%%%%%%%%%%%%%%%%%%%%%%%%%%%%%%%%%%%%%%%%%%%%%%%%%%%%%%%%%%%%%%%%%%%%
	\acknowledgements{We are grateful to Benjamin Yadin, Richard Moles and Thomas Veness for helpful discussions and the Physical Institute for Theoretical Hierarchy (PITH) for encouraging an investigation into this topic. BM acknowledges financial support from the Engineering and Physical Sciences Research Council (EPSRC) under the Doctoral Prize Grant (Grant No.~EP/T517902/1). LF acknowledges financial support from the Austrian Science Fund (FWF) through SFB BeyondC (Grant No.~F7102). BL is supported by Leverhulme Trust Research Project Grant (RPG-2018-213). D.G. acknowledges support from FQXI (grant no. RFP-IPW-1907)}
	%\bibliographystyle{unsrt}
	%\bibliography{Index}
	
	\appendix
	\section{proof of theorem \ref{thm:CHSH_break}}\label{app:proof_thm1}
	\begin{theorem*}
		Given observers {$A$ and $B$}, each with measurement choice $x_k\in\{0_k,1_k\}$ with outcomes $y_k\in\{-1,+1\}$ whose systems are distributed according to some  {quasi}probability  distribution $\tilde{P}_\Lambda$, then the quasiprobabilistic Bell inequality holds:
		\begin{align}
		\left|{E}(0_A,0_B)-{E}(0_A,1_B)+{E}(1_A,0_B)+{E}(1_A,1_B)\right|\leq\,2+\mathcal{N}(\tilde{P}_\Lambda),
		\end{align}
		\color{black}
		{where}
		{\begin{align}
			\N(\tilde{P}_\Lambda)\coloneq\begin{cases}
			\N_+(\tilde{P}_\Lambda) \qquad\text{ if }{E}(1_A,0_B)+{E}(1_A,1_B)<0,\\
			\N_-(\tilde{P}_\Lambda) \qquad\text{ else,}
			\end{cases}
			\end{align}
		{is a negativity witness, and }
			 $$\N_\pm(\tilde{P}_\Lambda):=\sum_{\lambda_A,\lambda_B}\left[2\pm\left(\inner{A}_{\lambda_A}^{1_A}\inner{B}_{\lambda_B}^{1_B}+\inner{A}_{\lambda_A}^{1_A}\inner{B}_{\lambda_B}^{0_B}\right)\right]\left(\Big{|}\tilde{P}_\Lambda(\lambda_A,\lambda_B)\Big{|}-\tilde{P}_\Lambda(\lambda_A,\lambda_B) \right).$$  }
	\end{theorem*}
	\begin{proof}
		{The first part of the proof follows Bell's 1971 derivation of the CHSH inequality \cite{bell2004speakable}.} For brevity in the proof we will just write $\tilde{P}_\Lambda$ as $P$.
		
		We start by rewriting the correlation function, 
		\begin{align}
		{E}(x_A,x_B):=&\sum_{y_A,y_B}y_A y_B\sum_{\lambda_A,\lambda_B}P_A\left(y_A|x_A,\lambda_A\right)P_B\left(y_B|x_B,\lambda_B\right){P}\left(\lambda_A,\lambda_B\right)\nonumber\\ 
		=&\sum_{\lambda_A,\lambda_B}\inner{A}_{\lambda_A}^{x_A}\inner{B}_{\lambda_B}^{x_B}\,{P}(\lambda_A,\lambda_B),\label{eq:correlation}
		\end{align}
		where $\inner{k}_{\lambda_k}^{x_k}${, defined in equation \eqref{eq:local_exp},} is the {$\lambda_k$-local} expectation value for observer $k$ performing measurement $x_k$. Starting with the following difference between correlation functions, 
		\begin{align}
		{E}&(0_A,0_B)-{E}(0_A,1_B)\nonumber\\ =&\sum_{\lambda_A,\lambda_B}\left(\inner{A}_{\lambda_A}^{0_A}\inner{B}_{\lambda_B}^{0_B}-\inner{A}_{\lambda_A}^{0_A}\inner{B}_{\lambda_B}^{1_B}\right){P}(\lambda_A,\lambda_B)\nonumber\\ 
		=&\sum_{\lambda_A,\lambda_B}\left(\inner{A}_{\lambda_A}^{0_A}\inner{B}_{\lambda_B}^{0_B}-\inner{A}_{\lambda_A}^{0_A}\inner{B}_{\lambda_B}^{1_B}\pm\inner{A}_{\lambda_A}^{0_A}\inner{B}_{\lambda_B}^{0_B}\inner{A}_{\lambda_A}^{1_A}\inner{B}_{\lambda_B}^{1_B}\mp\inner{A}_{\lambda_A}^{0_A}\inner{B}_{\lambda_B}^{0_B}\inner{A}_{\lambda_A}^{1_A}\inner{B}_{\lambda_B}^{1_B} \right){P}(\lambda_A,\lambda_B)\nonumber\\ 
		=&\sum_{\lambda_A,\lambda_B}\inner{A}_{\lambda_A}^{0_A}\inner{B}_{\lambda_B}^{0_B}\left(1\pm\inner{A}_{\lambda_A}^{1_A}\inner{B}_{\lambda_B}^{1_B}\right){P}(\lambda_A,\lambda_B)-\sum_{\lambda_A,\lambda_B}\inner{A}_{\lambda_A}^{0_A}\inner{B}_{\lambda_B}^{1_B}\left(1\pm\inner{A}_{\lambda_A}^{1_A}\inner{B}_{\lambda_B}^{0_B}\right){P}(\lambda_A,\lambda_B),\label{eq:same_sign}
		\end{align}
		{where the ``$\pm$'' in equation \eqref{eq:same_sign} is to be understood as either ``$+$'' in all terms or ``$-$'' in all terms.}
		Taking the absolute value of both sides and using the triangular inequality, 
		\begin{align}
		\big{|}{E}&(0_A,0_B)-{E}(0_A,1_B)\big{|}\nonumber\\ \leq&\Big{|}\sum_{\lambda_A,\lambda_B}\inner{A}_{\lambda_A}^{0_A}\inner{B}_{\lambda_B}^{0_B}\left(1\pm\inner{A}_{\lambda_A}^{1_A}\inner{B}_{\lambda_B}^{1_B}\right){P}(\lambda_A,\lambda_B)\Big{|}+\Big{|}\sum_{\lambda_A,\lambda_B}\inner{A}_{\lambda_A}^{0_A}\inner{B}_{\lambda_B}^{1_B}\left(1\pm\inner{A}_{\lambda_A}^{1_A}\inner{B}_{\lambda_B}^{0_B}\right){P}(\lambda_A,\lambda_B)\Big{|}.
		\label{equ:minus_absol}
		\end{align}
		Starting with {the first term on the right-hand side of inequality \eqref{equ:minus_absol}}, we again apply the triangular inequality,
		\begin{align}
		\Big{|}\sum_{\lambda_A,\lambda_B}\inner{A}_{\lambda_A}^{0_A}\inner{B}_{\lambda_B}^{0_B}\left(1\pm\inner{A}_{\lambda_A}^{1_A}\inner{B}_{\lambda_B}^{1_B}\right){P}(\lambda_A,\lambda_B)\Big{|}
		\leq&\sum_{\lambda_A,\lambda_B} \Big{|}\inner{A}_{\lambda_A}^{0_A}\inner{B}_{\lambda_B}^{0_B}\left(1\pm\inner{A}_{\lambda_A}^{1_A}\inner{B}_{\lambda_B}^{1_B}\right){P}(\lambda_A,\lambda_B) \Big{|}\nonumber\\ 
		=&\sum_{\lambda_A,\lambda_B}\Big{|}\inner{A}_{\lambda_A}^{0_A}\inner{B}_{\lambda_B}^{0_B}\Big{|}\Big{|}\left(1\pm\inner{A}_{\lambda_A}^{1_A}\inner{B}_{\lambda_B}^{1_B}\right){P}(\lambda_A,\lambda_B) \Big{|}.
		\end{align}
		As $y_k\in\{-1,+1\}$ we can say $\Big{|}\inner{k}_{\lambda_k}^{x_k}\Big{|}\leq1\,\,\forall\,k$, we can write, 
		\begin{align}
		\Big{|}\sum_{\lambda_A,\lambda_B}\inner{A}_{\lambda_A}^{0_A}\inner{B}_{\lambda_B}^{0_B}\left(1\pm\inner{A}_{\lambda_A}^{1_A}\inner{B}_{\lambda_B}^{1_B}\right){P}(\lambda_A,\lambda_B)\Big{|}\leq&\sum_{\lambda_A,\lambda_B}\Big{|}\left(1\pm\inner{A}_{\lambda_A}^{1_A}\inner{B}_{\lambda_B}^{1_B}\right){P}(\lambda_A,\lambda_B) \Big{|}\nonumber\\ 
		=&\sum_{\lambda_A,\lambda_B}\Big{|}\left(1\pm\inner{A}_{\lambda_A}^{1_A}\inner{B}_{\lambda_B}^{1_B}\right)\Big{|}\Big{|}{P}(\lambda_A,\lambda_B) \Big{|}\nonumber\\ 
		=&\sum_{\lambda_A,\lambda_B}\left(1\pm\inner{A}_{\lambda_A}^{1_A}\inner{B}_{\lambda_B}^{1_B}\right)\Big{|}{P}(\lambda_A,\lambda_B) \Big{|}
		\label{equ:simplifying}
		\end{align}
		where we have used the fact that $\left(1\pm\inner{A}_{\lambda_A}^{1_A}\inner{B}_{\lambda_B}^{1_B}\right)$ is necessarily non-negative because of the choice of eigenvalues, $y_k\in\{-1,+1\}$. 
		
		{Similarly, we find for the second term on the right-hand side of inequality~\eqref{equ:minus_absol}
			\begin{align}
			\Big{|}\sum_{\lambda_A,\lambda_B}\inner{A}_{\lambda_A}^{0_A}\inner{B}_{\lambda_B}^{1_B}\left(1\pm\inner{A}_{\lambda_A}^{1_A}\inner{B}_{\lambda_B}^{0_B}\right){P}(\lambda_A,\lambda_B)\Big{|}\leq &\sum_{\lambda_A,\lambda_B}\left(1\pm\inner{A}_{\lambda_A}^{1_A}\inner{B}_{\lambda_B}^{0_B}\right)\Big{|}{P}(\lambda_A,\lambda_B) \Big{|}.\label{equ:simplifying_2}
			\end{align}
			By adding inequalities \eqref{equ:simplifying} and \eqref{equ:simplifying_2} we find the following upper bound for the left-hand side of inequality \eqref{equ:minus_absol},
			\begin{align}
			\big{|}{E}(0_A,0_B)-{E}(0_A,1_B)\big{|} \leq&\sum_{\lambda_A,\lambda_B}\left[2\pm\left(\inner{A}_{\lambda_A}^{1_A}\inner{B}_{\lambda_B}^{1_B}+\inner{A}_{\lambda_A}^{1_A}\inner{B}_{\lambda_B}^{0_B}\right)\right]\Big{|}{P}(\lambda_A,\lambda_B) \Big{|}.
			\label{equ:minus_absol_total}
			\end{align}
			So far the proof followed Bell's 1971 derivation \cite{bell2004speakable} of the CHSH inequality. In Bell's derivation, one assumes that the joint probability distribution is positive, $P(\lambda_A,\lambda_B)\geq 0$, which, using the definition of the correlation function and the triangle inequality, leads to the well-known CHSH inequality, $\left|{E}(0_A,0_B)-{E}(0_A,1_B)+{E}(1_A,0_B)+{E}(1_A,1_B)\right|\leq\,2$.}
		
		{We have to take another approach because here $P(\lambda_A,\lambda_B)$ can be a quasiprobability distribution and thus take negative values. For each of the two inequalities \eqref{equ:minus_absol_total} (corresponding to the choice for ``$\pm$''), we define a negativity witness $\N_\pm(P)$ for some normalised distribution $P\in\tilde{\P}$ as the difference obtained by replacing $\left|P(\lambda_A,\lambda_B)\right|$ with $P(\lambda_A,\lambda_B)$ in the right-hand side of inequality \eqref{equ:minus_absol_total},
			\begin{align}
			\N_\pm(P):=\sum_{\lambda_A,\lambda_B}\left[2\pm\left(\inner{A}_{\lambda_A}^{1_A}\inner{B}_{\lambda_B}^{1_B}+\inner{A}_{\lambda_A}^{1_A}\inner{B}_{\lambda_B}^{0_B}\right)\right]\left[\left|P\left(\lambda_A,\lambda_B\right)\right|-P\left(\lambda_A,\lambda_B\right) \right].
			\end{align}
			Note that although this negativity witness is perfectly valid according to definition \ref{def:neg_witness}, it is not faithful because $2\pm\left(\inner{A}_{\lambda_A}^{1_A}\inner{B}_{\lambda_B}^{1_B}+\inner{A}_{\lambda_A}^{1_A}\inner{B}_{\lambda_B}^{0_B}\right)$ may be zero for $P\in \tilde{\P}/\P$, i.e., $\N_\pm(P)$ may be zero for a quasiprobability distribution. Nevertheless we can now write inequality \eqref{equ:minus_absol_total} as,
			\begin{align}
			\big{|}{E}(0_A,0_B)-{E}(0_A,1_B)\big{|} \leq&\,\sum_{\lambda_A,\lambda_B}\left[2\pm\left(\inner{A}_{\lambda_A}^{1_A}\inner{B}_{\lambda_B}^{1_B}+\inner{A}_{\lambda_A}^{1_A}\inner{B}_{\lambda_B}^{0_B}\right)\right]{P}(\lambda_A,\lambda_B) +\N_\pm(P).
			\label{equ:minus_absol_total_2}
			\end{align}
			The first term on the right-hand side of inequality \eqref{equ:minus_absol_total_2} can then be simplified using the definition of the correlation function \eqref{eq:correlation} and that ${P}(\lambda_A,\lambda_B)$ is normalized,
			\begin{align}
			&\sum_{\lambda_A,\lambda_B}\left[2\pm\left(\inner{A}_{\lambda_A}^{1_A}\inner{B}_{\lambda_B}^{1_B}+\inner{A}_{\lambda_A}^{1_A}\inner{B}_{\lambda_B}^{0_B}\right)\right]{P}(\lambda_A,\lambda_B)\notag \\
			&=\sum_{\lambda_A,\lambda_B}2{P}(\lambda_A,\lambda_B)\pm\sum_{\lambda_A,\lambda_B}\left(\inner{A}_{\lambda_A}^{1_A}\inner{B}_{\lambda_B}^{1_B}+\inner{A}_{\lambda_A}^{1_A}\inner{B}_{\lambda_B}^{0_B}\right){P}(\lambda_A,\lambda_B)\notag \\
			&=2\pm\left[{E}(1_A,0_B)+{E}(1_A,1_B)\right].
			\end{align} 
			Thus, inequality \eqref{equ:minus_absol_total_2} becomes
			\begin{align}
			\big{|}{E}(0_A,0_B)-{E}(0_A,1_B)\big{|} \leq&\,2\pm\left[{E}(1_A,0_B)+{E}(1_A,1_B)\right]+\N_\pm(P).
			\label{equ:minus_absol_total_tight_3}
			\end{align}
			Now, we choose the inequality corresponding to ``$+$'' if $\left[{E}(1_A,0_B)+{E}(1_A,1_B)\right]$ is negative, and the inequality corresponding to ``$-$'' else.  This allows us to write
			\begin{align}
			\big{|}{E}(0_A,0_B)-{E}(0_A,1_B)\big{|} \leq&\,2-\left|{E}(1_A,0_B)+{E}(1_A,1_B)\right|+\N(P), \label{eq:almost_there}
			\end{align}
			where we defined 
			\begin{align}
			\N(P)\coloneq\begin{cases}
			\N_+(P) \qquad\text{ if }{E}(1_A,0_B)+{E}(1_A,1_B)<0,\\
			\N_-(P) \qquad\text{ else.}
			\end{cases}
			\end{align} 
			From inequality \eqref{eq:almost_there}, we obtain
			\begin{align}
			\big{|}{E}(0_A,0_B)-{E}(0_A,1_B)\big{|} +\left|{E}(1_A,0_B)+{E}(1_A,1_B)\right|\leq&\,2+\N(P),
			\end{align}
			and with one final use of the triangular inequality we find a CHSH-type inequality for arbitrary $P\in \tilde P$,
			\begin{align}
			\left|{E}(0_A,0_B)-{E}(0_A,1_B) +{E}(1_A,0_B)+{E}(1_A,1_B)\right|\leq&\,2+\N(P),
			\end{align}
			completing the proof.}
		
	\end{proof}
	\section{proof of theorem \ref{thm:CHSH_break_N}}\label{app:proof_thm2}
	\begin{theorem*}
		Given observers {$A$ and $B$}, each with $n\geq2$ measurements 
		$x_k\in\{0_k,1_k,\dots,{n-1}_k\}$ with outcomes $y_k\in\{-1,+1\}$ whose systems are distributed according to some  {quasi}probability distribution $\tilde{P}_\Lambda$, 
		{\begin{align}\label{eq:CHSH_break_N_app}
			\left|\sum^{n-1}_{i=0}{E}(i_A,i_B)+\sum^{n-1}_{i=1}{E}(i_A,i-1_B)-{E}(0_A,n-1_B)\right|\leq\,2n-2+{\mathcal{N}_n({\tilde{P}_\Lambda})}, 
			\end{align}}
		where { $\mathcal{N}_n({\tilde{P}_\Lambda})=\sum^{n-1}_{i=1}\N^{(i)}(\tilde{P}_\Lambda)$ {is a negativity witness} with}
		\begin{align}
		\N^{(x)}(\tilde{P}_\Lambda)\coloneq\begin{cases}
		\N^{(x)}_+(\tilde{P}_\Lambda) \qquad\text{ if }{E}(0_A,x_B)+{E}(0_A,x-1_B)<0,\\
		\N^{(x)}_-(\tilde{P}_\Lambda) \qquad\text{ else,}
		\end{cases}
		\end{align}
		where $\N^{(x)}_{\pm}({\tilde{P}_\Lambda})\coloneqq\sum_{\lambda_A,\lambda_B}\left[2\pm\left(\inner{A}_{\lambda_A}^{x_A}\inner{B}_{\lambda_B}^{x_B}+\inner{A}_{\lambda_A}^{x_A}\inner{B}_{\lambda_B}^{x-1_B}\right)\right]\left(\Big{|}{\tilde{P}_\Lambda}(\lambda_A,\lambda_B)\Big{|}-{\tilde{P}_\Lambda}(\lambda_A,\lambda_B) \right)$.
	\end{theorem*}
	{\begin{proof}
			The proof is similar to the creation of chained CHSH inequalities, see \cite{braunstein1990wringing} for an intuitive description, and works by induction in $n$.\\
			{\textit{Anchor step}} $n=2$: see theorem \ref{thm:CHSH_break}.\\
			\textit{Inductive step}: Suppose that theorem \ref{thm:CHSH_break_N} holds for $n=k$. We will prove the theorem for $n=k+1$. Starting from the left-hand side of equation \eqref{eq:CHSH_break_N_app} for $n=k+1$, we find
			\begin{align}
			&\left|\sum^{k}_{i=0}{E}(i_A,i_B)+\sum^{k}_{i=1}{E}(i_A,i-1_B)-{E}(0_A,k_B)\right|\\
			&=\left|\sum^{k-1}_{i=0}{E}(i_A,i_B)+\sum^{k-1}_{i=1}{E}(i_A,i-1_B)-{E}(0_A,k_B)+{E}(k_A,k_B)+{E}(k_A,k-1_B)\right|\\
			&=\left|\sum^{k-1}_{i=0}{E}(i_A,i_B)+\sum^{k-1}_{i=1}{E}(i_A,i-1_B)-{E}(0_A,k_B)+{E}(k_A,k_B)+{E}(k_A,k-1_B)\right.\notag\\&\qquad+{E}(0_A,k-1_B)-{E}(0_A,k-1_B)\Biggr|\\
			&\leq\left|\sum^{k-1}_{i=0}{E}(i_A,i_B)+\sum^{k-1}_{i=1}{E}(i_A,i-1_B)-{E}(0_A,k-1_B)\right|+\notag\\&\qquad \left|{E}(k_A,k_B)+{E}(k_A,k-1_B)+{E}(0_A,k-1_B)-{E}(0_A,k_B)\right|\label{eq:induction_ineq}\\
			&\leq2k-2+\sum^{k-1}_{i=1}\N^{(i)}(\tilde{P}_\Lambda)+2+\N^{(k)}(\tilde{P}_\Lambda)\\
			&=2(k+1)-2+\N_{k+1}(\tilde{P}_\Lambda),
			\end{align}
			which concludes the induction. The inequality in line \eqref{eq:induction_ineq} is the triangle inequality, and we proceed from that line by using the induction hypothesis and theorem \ref{thm:CHSH_break} for measurements $0_A$, $k_A$ for Alice, and $k-1_B$, and $k_B$ for Bob.
	\end{proof}}
	\section{Saturation of {the} $n$-measurement quasiprobabilistic Bell inequality}\label{app:sat_n_ex}
	We can generalise the $2$-measurement example from the main text to $n$ measurements in the following way. {Using equation \eqref{eq:correlation}, we rewrite the left-hand side of equation  \eqref{eq:CHSH_break_N} as,
	\begin{align}
\Bigg{|}\sum_{\lambda}\M(\lambda)\,\tilde{P}_\Lambda(\lambda)\Bigg{|}, 
	\end{align}
where we use only a single hidden variable $\lambda$, and 
\begin{align}
\M (\lambda):=\sum^{n-1}_{i=0}\inner{A}_{\lambda}^{i_A}\inner{A}_{\lambda}^{i_B}+\sum^{n-1}_{i=1}\inner{A}_{\lambda}^{i_A}\inner{A}_{\lambda}^{i-1_B}-\inner{A}_{\lambda}^{0_A}\inner{A}_{\lambda}^{n-1_B}
\end{align}
are the scores of each of the $\lambda$-local distributions, that is  $-(2n-2)\leq \M(\lambda)\leq 2n-2$ {holds}.}

We again consider $4$ classical scenarios, $3$ of which achieve {a score of} $2n-2$ but now the last achieving {a score of} $2n-6$.  The source produces each of the distributions according to the following quasiprobability distribution,
        \begin{align}
     \tilde{P}_\Lambda(\lambda)=\begin{cases}
 	\frac{4+\mathcal{N}_n}{12}\qquad\text{ for }\lambda=1,2,3,\\
 	-\frac{\mathcal{N}_n}{4} \qquad\,\,\,\text{  for }\lambda=4,
 		\end{cases}\label{eq:quasi_exmaple_n}
     \end{align}
where $\lambda=1,2,3$ corresponds to classical distributions with {score} $2n-2$, and $\lambda=4$ to $2n-6$. We can see that this distribution saturates the $n$-measurement quasiprobabilistic Bell inequality from theorem \ref{thm:CHSH_break_N}, 
\begin{align}\label{equ:quasi_n_measurement}
    (2n-2)\tilde{P}_\Lambda(1) +(2n-2)\tilde{P}_\Lambda(2)+  (2n-2)\tilde{P}_\Lambda(3)+(2n-6)\tilde{P}_\Lambda(4)=2n-2+\mathcal{N}_n.
\end{align}
We now need to come up with the $\lambda$-local {probability} distributions which result in a well-defined $P\left(y_A,y_B|x_A,x_B\right)$ and gives the correct value for the witness $\mathcal{N}_n$. To do this we can generalise the classical distributions from the main text for $n$ measurements, using the same notation as previously, such classical distributions are, 
\begin{align}
    [(\overbrace{-,\dots,-}^{n})_A,(\overbrace{-,\dots,-}^{n-1},+)_B]^{\lambda=1},& [(+,\overbrace{-,\dots,-}^{n-1})_A,(\overbrace{-,\dots,-}^{n})_B]^{\lambda=2}, [(\overbrace{+,\dots,+}^{n})_A,(\overbrace{+,\dots,+}^{n})_B]^{\lambda=3},\nonumber\\
    &[(+,\overbrace{-,\dots,-}^{n-1})_A,(\overbrace{-,\dots,-}^{n-1},+)_B]^{\lambda=4}.\label{equ:n-measuremnt_example}
\end{align}
It {is easy to check} that all such distributions  achieve for $\lambda=1,2,3$ {a score}  $2n-2$, and {for} $\lambda=4$,  $2n-6$.
{Since the distribution for $\lambda=4$ enters into the total probability distribution, $P\left(y_A,y_B|x_A,x_B\right)$, with negative weight, the other $\lambda$-local distributions (with $\lambda=1,2,3$) must compensate for that negativity to ensure that the total probability distribution is valid.

It is easy to see that this is indeed the case by observing that for each combination of Alice and Bob's signs for $\lambda=4$ that same combination of symbols appear in the same places for at least one of the other distributions.}
We also find that requiring positivity of the total probability distribution also gives us the no-signalling condition:

\begin{align}
    \tilde{P}_\Lambda(\lambda)+ \tilde{P}_\Lambda(4)\geq0\text{ for } \lambda=1,2,3 \implies\N_n\leq2 \quad \forall n\geq2.  
\end{align}    
The final thing to check is that said distributions in equation \eqref{equ:n-measuremnt_example} coupled with the quasiprobability distribution in equation \eqref{eq:quasi_exmaple_n}  gives the required {value $\mathcal{N}_n$} for the {negativity} witness {$\mathcal{N}_n({\tilde{P}_\Lambda})=\sum^{n-1}_{i=1}\N^{(i)}(\tilde{P}_\Lambda)$ with}
		\begin{align}
		\N^{(x)}(\tilde{P}_\Lambda)\coloneq\begin{cases}
		\N^{(x)}_+(\tilde{P}_\Lambda) \qquad\text{ if }{E}(0_A,x_B)+{E}(0_A,x-1_B)<0,\\
		\N^{(x)}_-(\tilde{P}_\Lambda) \qquad\text{ else,}
		\end{cases}
		\end{align}
		where $\N^{(x)}_{\pm}({\tilde{P}_\Lambda})\coloneqq\sum_{\lambda_A,\lambda_B}\left[2\pm\left(\inner{A}_{\lambda_A}^{x_A}\inner{B}_{\lambda_B}^{x_B}+\inner{A}_{\lambda_A}^{x_A}\inner{B}_{\lambda_B}^{x-1_B}\right)\right]\left(\Big{|}{\tilde{P}_\Lambda}(\lambda_A,\lambda_B)\Big{|}-{\tilde{P}_\Lambda}(\lambda_A,\lambda_B) \right)$. 
		
Firstly, we can see by going through the distributions in equation \eqref{equ:n-measuremnt_example} that for all measurement choices $x$, ${E}(0_A,x_B)+{E}(0_A,x-1_B)>0$ meaning that the witness we calculate for all $x$ in the sum of $\mathcal{N}_n({\tilde{P}_\Lambda})$ is $\N^{(x)}_-(\tilde{P}_\Lambda)$. We then go through the expectation values in the definition of $\N^{(x)}_-(\tilde{P}_\Lambda)$ for all $x$ for the $\lambda=4$ distribution given in equation \eqref{equ:n-measuremnt_example}, from which we can see,
\begin{align}
2-\left(\inner{A}_{4}^{x_A}\inner{B}_{4}^{x_B}+\inner{A}_{4}^{x_A}\inner{B}_{4}^{x-1_B}\right)=\begin{cases}
 	0\qquad\text{ for }x=1,\dots,n-2\\
 	2 \qquad\,\,\,\text{  for }x=n-1.
 		\end{cases}
\end{align}
meaning that upon calculating $\mathcal{N}_n({\tilde{P}_\Lambda})\sum^{n-1}_{i=1}\N^{(i)}(\tilde{P}_\Lambda)$, we get
\begin{align}
    \mathcal{N}_n({\tilde{P}_\Lambda})=&2\left(\Big{|}{\tilde{P}_\Lambda}(4)\Big{|}-{\tilde{P}_\Lambda}(4) \right)\nonumber\\ 
    =&\mathcal{N}_n,
\end{align}
as required.

\end{document}